\documentclass[11pt,a4paper,reqno]{amsart} 

\usepackage[utf8]{inputenc}
\usepackage{hyperref}
\hypersetup{
  colorlinks   = true, 
  urlcolor     = blue, 
  linkcolor    = blue, 
  citecolor   = blue 
}

\usepackage{amsmath}
\usepackage{amsthm}
\usepackage{amsfonts}
\usepackage{amssymb}
\usepackage{array}
\usepackage[margin=2.6cm]{geometry} 
\usepackage{enumitem}

\usepackage{longtable}
\usepackage[dvipsnames]{xcolor}
\usepackage{colortbl}
\usepackage{tikz}
\usetikzlibrary{positioning}
\usepackage{graphicx}
\usepackage{xypic}

\usepackage{faktor}
\theoremstyle{definition}
\newtheorem{theorem}{Theorem}[section]

\newtheorem{proposition}[theorem]{Proposition}
\newtheorem{corollary}[theorem]{Corollary}
\newtheorem{lemma}[theorem]{Lemma}

\newtheorem{definition}[theorem]{Definition}
\newtheorem{example}[theorem]{Example}
\newtheorem{remark}[theorem]{Remark}

\newcommand{\K}{\mathbb K}
\newcommand{\F}{\mathbb F}
\newcommand{\E}{\mathbb E}

\newcommand{\N}{\mathbb N}
\newcommand{\LL}{\mathbb L}

\newcommand{\Fm}{\mathbb F_{q^m}}

\newcommand{\C}{\mathcal C}
\newcommand{\mS}{\mathcal S}
\newcommand{\mB}{\mathcal B}
\newcommand{\mU}{\mathcal U}
\newcommand{\mR}{\mathcal R}
\newcommand{\Ev}{\mathrm{Ev}}

\newcommand{\bfn}{\mathbf{n}}

\DeclareMathOperator{\wt}{wt}
\DeclareMathOperator{\rk}{rk}
\DeclareMathOperator{\SR}{srk}
\DeclareMathOperator{\dd}{d}
\DeclareMathOperator{\HH}{H}
\DeclareMathOperator{\XX}{x}
\DeclareMathOperator{\SX}{sx}
\DeclareMathOperator{\RR}{rk}
\DeclareMathOperator{\Gal}{Gal}
\DeclareMathOperator{\GL}{GL}
\DeclareMathOperator{\Aut}{Aut}
\DeclareMathOperator{\rowsp}{rowsp}
\DeclareMathOperator{\ord}{ord}
\DeclareMathOperator{\LIS}{LIS}
\DeclareMathOperator{\diag}{diag}

\title{Sum-rank product codes and bounds on the minimum distance}
\author[G. N. Alfarano]{Gianira N. Alfarano$^1$}
\address{$^1$University of Zurich, Switzerland}
\curraddr{}
\email{gianiranicoletta.alfarano@math.uzh.ch}

\author[F. J. Lobillo]{F. J. Lobillo$^2$}
\address{$^2$University of Granada, Spain}
\curraddr{}
\email{jlobillo@ugr.es}

\author[A. Neri]{Alessandro Neri$^3$}
\address{$^3$Max-Planck-Institute for Mathematics in the Sciences, Leipzig}
\curraddr{}
\email{alessandro.neri@mis.mpg.de}

\author[A. Wachter-Zeh]{Antonia Wachter-Zeh$^4$}
\address{$^4$Technical University of Munich, Germany}
\curraddr{}
\email{antonia.wachter-zeh@tum.de}

\begin{document}

\maketitle

\begin{abstract}
    The tensor product of one code endowed with the Hamming metric and one endowed with the rank metric is analyzed. This gives a code which naturally inherits the sum-rank metric. Specializing to the product of a cyclic code and a skew-cyclic code, the resulting code turns out to belong to the recently introduced family of cyclic-skew-cyclic. A group theoretical description of these codes is given, after investigating the semilinear isometries in the sum-rank metric. Finally, a generalization of the Roos and the Hartmann-Tzeng bounds for the sum rank-metric is established, as well as a new lower bound on the minimum distance of one of the two codes constituting the product code.
\end{abstract}

\section{Introduction}

Due to their rich algebraic structure, \emph{cyclic codes} represent undoubtedly one of the most studied families of linear codes. Indeed, their polynomial interpretation as ideals of $\F[z]/(z^n-1)$ is central to the efficiency of encoding and decoding procedures, and it also allows to design codes with desired properties.
Among these properties, several bounds on the minimum distance of cyclic codes have been established by using sets of consecutive elements in the \emph{defining set}, which is the collection of the zeros of all the polynomials in the code. Out of the most famous bounds, we can  certainly find  the BCH \cite{bose1960class, bose1960further, hocquenghem1959codes}, the Hartmann-Tzeng (HT) \cite{hartmann1972generalizations}  and the Roos bound \cite{roos1983new, roos1982generalization}. 
More recently, generalizations of these types of bounds have been proposed by embedding the code into a cyclic product code; see \cite{zeh2013generalizing}. The idea of \cite{zeh2013generalizing} is that one can find a lower bound on the minimum distance of a cyclic code by using a second auxiliary cyclic code and apply the HT bound to the minimum distance of the obtained product cyclic code. The bound obtained in this way is more general, and the  HT bound for cyclic codes represents only a special instance. 

Starting from the work of Boucher, Geiselmann and Ulmer  \cite{boucher2007skew}, the notion of cyclicity has been then extended to codes defined over a skew polynomial ring; see also  \cite{boucher2009codes, boucher2009coding, chaussade2009skew}. In particular, in \cite{chaussade2009skew} a BCH bound for skew-cyclic codes has been generalized, which works for both the Hamming and the rank metric.
Later, in \cite{Gomez/Lobillo/Navarro/Neri:2018} the skew version of the HT bound has been proposed and in \cite{alfarano2021roos} a Roos-like bound for skew-cyclic codes in the Hamming and in the rank metric has been established. 
Moreover, these results confirmed that the rank metric is the natural measure inherited from the skew-polynomial structure of skew-cyclic codes.

Recently, codes endowed with the sum-rank metric have become popular due to their use in different applications; see for example \cite{nobrega2010multishot, martinez2019reliable, lu2005unified, martinez2019universal, mahmood2016convolutional}. From the theoretical point of view, the sum-rank metric is a generalization of both the Hamming metric and the rank metric. 
Formally, given a field extension $\F/\E$, the sum-rank metric is defined on $\F^n$, with respect to a fixed partition $(n_1,\ldots,n_{\ell})$ of length $n$. More precisely, the sum-rank of a vector {$v=(v^{(1)},\ldots, v^{(\ell)})$}, where $v^{(i)}\in\F^{n_i}$, is defined as the sum of the rank weights of the $v_i$ with respect to the extension $\F/\E$. When $\ell=1$, the sum-rank metric coincides with the rank metric on $\F^{n}$, and when $n_1=\ldots=n_{\ell}=1$ it coincides with the Hamming metric on $\F^{\ell}$. In 2020, Mart\'inez-Pe\~{n}as introduced the family of \emph{cyclic-skew-cyclic} codes endowed with the sum-rank metric; see \cite{martinez2020sum}. These codes coincide with cyclic codes when the sum-rank metric recovers the Hamming metric and with skew-cyclic codes when the sum-rank metric recovers the rank metric. 

\medskip 

Inspired by the work of  Mart\'inez-Pe\~{n}as \cite{martinez2020sum} and by the bounds found in \cite{zeh2013generalizing}, we consider the tensor product of a cyclic code endowed with the Hamming metric with a skew-cyclic code endowed with the rank-metric, both defined over the same  field $\F$. Such a product code turns out to be a cyclic-skew-cyclic code, which naturally inherits the sum-rank metric. For this reason, we first provide generalizations of the Roos bound and the Hartmann-Tzeng bound for cyclic-skew-cyclic codes endowed with the sum-rank metric. This was listed among the open problems left in \cite{martinez2020sum} as a result of independent interest. Furthermore, we study the  tensor product mentioned above and use the  generalized Roos and Hartmann-Tzeng bounds in order to derive new lower bounds  on the minimum Hamming distance of the  constituent cyclic code and on the minimum rank distance of the constituent skew-cyclic code.

\medskip

The paper is structured as follows. Section \ref{sec:preliminaries} contains a brief recap on the notions and results needed for the rest of the paper. There, we recall the Hamming, rank and sum-rank metric, as well as what is a cyclic-skew-cyclic code and its representation in a skew polynomial ring.  
In Section \ref{sec:isometries} we investigate the semilinear isometries for the sum-rank metric with respect to any partition of the length. This is needed for providing a group-theoretical description of a cyclic-skew-cyclic code, which is a property invariant under the equivalence relation defined by the semilinear isometries. Section \ref{sec:Roos_HT_bounds} focuses on lower bounds for the minimum sum-rank distance of a cyclic-skew-cyclic code. We provide  the analogues of the  Roos bound and of the  Hartmann-Tzeng bound, exploiting the notion of defining sets for cyclic-skew-cyclic codes. In Section \ref{sec:product_codes} we focus on tensor products of codes. We combine a cyclic code in the Hamming metric with a skew-cyclic code in the rank metric, and study the resulting product code with respect to the sum-rank metric. Finally, in Section \ref{sec:product_bounds} we combine the results on tensor product of codes with the Roos and the Hartmann-Tzeng bounds in order to obtain new lower bounds on the minimum distances of the constituent codes.

\section{Preliminaries and Setting}\label{sec:preliminaries}
In this section we fix the framework  that we are going to use in the rest of the paper and we recall the main tools and mathematical background needed for our purposes.

Let $m,s$ be positive integers, let $\E$ be a field and consider $\K/\E$ and $\F/\E$  extension fields of finite degree, respectively $h$ and $m$, such that $\K\cap\F=\E$. Moreover, let $\LL:=\F\K$. For convenience of the reader we will summarize the fields containment in Figure \ref{fig:diagram}, which will be recalled whenever it is necessary in the rest of the paper.
\begin{figure}[!h]
\begin{center}
\begin{tikzpicture}[node distance = 2cm, auto]
      \node (K) {$\LL$};
      \node (F) [below left of=K] {$\F$};
      \node (G) [below right of=K] {$\K$};
      \node (E) [below right of=F] {$\E$};
      \draw[-] (K) to node [swap] {\small{$h$}} (F);
      \draw[-] (F) to node [swap] {\small{$m$}} (E);
      \draw[-] (K) to node  {\small{$m$}} (G);
      \draw[-] (G) to node  {\small{$h$}} (E);
      
      \end{tikzpicture}
    
\end{center}
\caption{Field extensions and their degree.}\label{fig:diagram}
\end{figure}
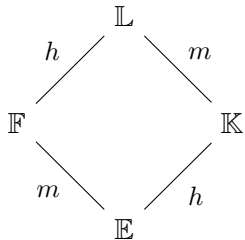
Finally, let $\ell$ be a positive integer and assume that $x^\ell-1$ splits into linear factors in $\K$, i.e. $\K$ contains all the $\ell$-th roots of unity.

Note that throughout the paper, the field $\F$ will essentially always be the defining field for our linear codes. Nevertheless, we will also make a wide use of the other auxiliary fields, in order to control parameters and properties of such codes.

\subsection{The Hamming and the Rank metric}
The \emph{Hamming weight}  on $\F^\ell$ is defined as 
$$ \begin{array}{rccl}
\wt_{\HH}: & \F^\ell & \longrightarrow & \N \\
&c&\longmapsto&  |\{i : c_i\ne 0\}|,\end{array}$$
i.e., $\wt_{\HH}(c)$ is the number of coordinates in $c$ that are different from 0. 

The Hamming weight induces a metric on $\F^\ell$, which is defined as $\dd_{\HH}(u,v):=\wt_{\HH}(u-v)$ for any $u,v\in\F^\ell$.
An $[\ell,k,d]_\F$ (\emph{Hamming metric}) \emph{code} $\C$  is a $k$-dimensional $\F$-subspace of $\F^\ell$ and $d=d_{\HH}(\C) := \min\{\dd_{\HH}(u,v) \mid u,v\in \C, u\ne v\}$. If $d$ is not known, we simply write $[\ell,k]_\F$ code. The parameter $\ell$ is the \emph{length} of $\C$, $k$ is its \emph{dimension} and $d$ is called \emph{minimum Hamming distance} of $\C$.
 
Another well-known metric is the rank metric. Let $\F/\E$ be the field extension defined above and let $N$ be a positive integer. The \emph{rank weight} for $\F/\E$ is defined as the following map:

$$ \begin{array}{rccl}
\wt_{\RR}^{\F/\E}: & \F^N & \longrightarrow & \N \\
&c&\longmapsto&  \dim_\E(\langle c_0,\ldots,c_{N-1} \rangle_\E).\end{array}$$
Also in this case, the rank weight $\wt_{\RR}^{\F/\E}$ induces a metric on the space $\F^N$ with respect to the field extension $\F/\E$, known as \emph{rank distance} and defined as $\dd_{\RR}^{\F/\E}(u,v):=\wt_{\RR}^{\F/\E}(u-v)$ for any $u,v\in\F^N$. Furthermore, we define an $[N, k, d]_{\F/\E}$ \emph{rank-metric code} $\C$ as a $k$-dimensional $\F$-subspace of $\F^N$, endowed with the rank metric. As for the Hamming metric case, the integer $N$ is called the \emph{length} of $\C$, $k$ is the \emph{dimension} of $\C$ and $d$ is defined as
$$d = \dd_{\RR}^{\F/\E}(\C) :=\min\{d_{\RR}^{\F/\E}(u, v) \mid u, v \in \C, u\ne v\}$$ and is called the \emph{minimum rank distance} of $\C$. If $d$ is not known we simply write that $\C$ is an $[N,k]_{\F/\E}$ rank-metric code.

Clearly, one may also define the rank metric for other extensions. In this paper we will only consider it with respect to $\F/\E$ and to $\LL/\K$. It will always be clear from the context which extension we will be considering, and, in order to lighten the notation, from now on we will omit the field extension  from the notation of rank weight and rank distance, writing only $\wt_{\RR}$ and $\dd_{\RR}$, respectively.

\subsection{The sum-rank metric}

In this subsection, we recall another metric, which is the natural combination of the Hamming and the rank ones. We will use the same field extensions defined in Figure \ref{fig:diagram}.

Consider the field extension $\F/\E$ and choose a partition   of $n$ as $n=n_1+\ldots+n_{\ell}$, which we denote by $\mathbf{n} = (n_1,\ldots,n_{\ell})$.
For a given vector $c\in \F^n$, this induces a partition of $c$ as
\begin{equation}\label{eq:vector_partition}c=(c^{(0)} \mid \ldots \mid c^{(\ell-1)}),\end{equation}
where $c^{(i)}=(c^{(i)}_0,\ldots,c^{(i)}_{n_i-1})\in \F^{n_i}$, for every $i \in \{0,\dots, \ell-1\}$. 
With this notation in mind, we define the \emph{sum-rank weight} for $\F/\E$ with respect to the partition $\mathbf{n}=(n_1,\ldots,n_{\ell})$ of $n$ as the function
$$ \begin{array}{rccl}
\wt_{\SR}^{\mathbf{n}, \F/\E}: & \F^n & \longrightarrow & \N \\
&c&\longmapsto& 
\sum\limits_{i=0}^{\ell-1}\wt_{\RR}(c^{(i)}).\end{array}$$
The sum-rank weight induces a metric on $\F^n$ defined as $\dd_{\SR}^{\mathbf{n}, \F/\E}(u,v):=\wt_{\SR}^{\mathbf{n}, \F/\E}(u-v)$ for any $u,v \in \F^n$, which we  call  \emph{sum-rank distance}.

Clearly, as for the rank metric, the same holds if we replace $\F/\E$ with $\LL/\K$. In this case we can define the sum-rank weight with respect to $\LL/\K$, and we denote it as
$$\wt_{\SR}^{\mathbf{n}, \LL/\K} : \LL^n \longrightarrow \N.$$
In the rest of the paper, these two extensions will be the only ones considered when computing the sum-rank metric. Moreover, from now on, whenever the field extension is known, it will be omitted from the notation of sum-rank weight and sum-rank distance, and we will only write $\wt_{\SR}^{\mathbf{n}}$ and $\dd_{\SR}^{\mathbf{n}}$.

Note that sum-rank distance recovers the Hamming distance and the rank distance by setting $n_1=\ldots=n_{\ell}=1$ and $\ell=1$, respectively.

\begin{definition}
Let $\F/\E$ be an arbitrary field extension and fix a partition of $n$ to be $\mathbf{n}=(n_1,\dots,n_r)$.
A (\emph{sum-rank metric}) \emph{code} $\C$ is an $\F$-linear subspace of $\F^n$ endowed with the sum-rank metric. The \emph{minimum sum-rank distance} is defined as usual as $\dd_{\SR}^\mathbf{n}(\C)=\min\{\wt_{\SR}^\mathbf{n}(c) \mid c\in \C, c\ne 0\}$ or equivalently $\dd_{\SR}^\mathbf{n}(\C)=\min\{\dd_{\SR}^\mathbf{n}(u,v) \mid u,v \in \C, u\ne v\}$, where $\dd_{\SR}^\mathbf{n}$ and $\wt_{\SR}^\mathbf{n}$ are defined with respect to $\F/\E$ and to the partition $\mathbf{n}$.
\end{definition}

In this work, the sum-rank metric codes that we are going to define have underlying field $\F$ and the fixed partition will always be $(\underbrace{N, \ldots, N}_{\ell\textrm{ times}})$, except for Section \ref{sec:isometries}, in which we deal  with a general partition $\mathbf{n}$.

\subsection{Skew polynomial rings}
We briefly recall the definition of skew polynomials rings. We refer to the seminal paper of Ore \cite{ore1933theory}, where  skew polynomials over division rings have been introduced.

Consider the field extension $\LL/\K$ and its Galois group $\Gal(\LL/\K)$. Let $\sigma \in \Gal(\LL/\K)$ be an automorphism of $\LL$.
We define the skew polynomial ring as the ring $\LL[z;\sigma]$ induced by $\sigma$ over $\LL$. The multiplication rule over $\LL[z;\sigma]$ is given by $z^iz^j = z^{i+j}$ and $za=\sigma(a)z $, for all $a\in \LL$. Since $\LL$ is a field,  $\LL[z;\sigma]$ is in particular a left and right Euclidean domain. 


We also recall the evaluation of skew polynomials. A systematic approach to this concept was started in 1986 by Lam and Leroy. We refer the interested reader to \cite{lam1988vandermonde, lam2004wedderburn, lam2008wedderburn,Delenclos/Leroy:2007}. Here we are only interested in the following definition. Let $f(z)\in\LL[z;\sigma]$ and $a\in \LL$. We define the evaluation of $f(z)$ in $a$ as the unique element $f(a)$ such that 
\begin{equation}\label{eq:righteval}
    f(z)=q(z)(z-a)+f(a),
\end{equation}
where $q(z)\in\LL[z;\sigma]$.

\subsection{Cyclic-skew-cyclic codes}
In this subsection we give the definition of cyclic-skew-cyclic codes in their vector and polynomial representations for arbitrary fields. The results presented here can be found for finite fields in \cite{martinez2020sum}.

Assume that $\LL/\K$ is a cyclic Galois extension and let  $\sigma \in \Gal(\LL/\K)$.  Moreover, denote by $\theta$  the restriction $\sigma_{\mid_\F}$ of $\sigma$ to $\F$.
Recall that $\ell$ is a positive integer, such that the roots of $x^\ell-1$  belong to $\K$ and they are pairwise distinct.

We fix the setting of the previous subsection, choosing $n=\ell N$, and fixing the partition $\mathbf{n} =(N,\ldots,N)$ of $n$. Accordingly,  we partition every vector $c\in \F^n$ as in \eqref{eq:vector_partition}.

Define the \emph{block-shift operator} $\rho$ on $\F^n$ as
\begin{equation}\label{eq:rho}\rho((c^{(0)} \mid \ldots \mid c^{(\ell-1)}))=(c^{(\ell-1)}\mid c^{(0)}  \mid \ldots \mid c^{(\ell-2)}),\end{equation}
and the \emph{$\theta$-inblock shift operator} $\phi$ on $\F^n$ as
\begin{equation}\label{eq:phi}\phi((c^{(0)} \mid \ldots \mid c^{(\ell-1)}))=(\varphi(c^{(0)}) \mid \ldots \mid \varphi(c^{(\ell-1)})),\end{equation}
where $\varphi:\F^N\rightarrow \F^N$ is given by
$$ \varphi(v_0,\ldots,v_{N-1})=(\theta(v_{N-1}),\theta(v_0),\ldots, \theta(v_{N-2})).$$

With the setting fixed above, we recall the definition of cyclic-skew-cyclic codes, introduced in \cite{martinez2020sum}.

\begin{definition}
A code $\C\subseteq \F^n$ is called \emph{cyclic-skew-cyclic} if $\rho(\C)\subseteq \C$ and $\phi(\C)\subseteq \C$. These are actually equalities.
\end{definition}

\subsection{Skew polynomial representation of cyclic-skew-cyclic codes}

As above, assume that $\LL/\K$ is a cyclic Galois extension of degree $m$. Let  $\sigma \in \Gal(\LL/\K)$ and let $\theta$  be the restriction $\sigma_{\mid_\F}$ of $\sigma$ to $\F$. Moreover, consider the skew polynomial ring $\F[z;\theta]$.
 Let $t$ be the order of $\theta$. It is well-known that the center of $\F[z;\theta]$ is $\F^\theta[z^{t}]$, where $\F^\theta$ denotes the subfield of elements which are fixed by $\theta$. Moreover, assume $t=m$, and $\F^\theta=\E$
 
Moreover, we also assume that $m$ divides $N$. Then, $z^N-1$ belongs to the center of $\F[z;\theta]$ and therefore it generates a two-sided ideal. Define 
\begin{align*}
   \mS &:= \faktor{\F[z;\theta]}{(z^N-1)}, \\
    \mR &:=\faktor{\mS[x]}{(x^\ell-1)},
\end{align*}
where it is easy to see that also $(x^\ell-1)$ is a two-sided ideal of $\mS[x]$.

At this point there is a natural identification of $\F^n$ with $\mR$, given by the map $\mu : \F^n  \longrightarrow  \mR,$ such that
$$ \mu((c^{(0)} \mid \ldots \mid c^{(\ell-1)}))= \sum_{i=0}^{\ell-1} \bigg(\sum_{j=0}^{N-1} c_j^{(i)}z^j \bigg)x^i. $$

One can see that it is also possible to interchange the order of the variables and consider another polynomial representation for $\F^n$. Define
\begin{align*}
   \mS' &:= \faktor{\F[x]}{(x^\ell-1)}, \\
    \mR' &:=\faktor{\mS'[z;\theta]}{(z^N-1)},
\end{align*}
where we have extended $\theta$ to  $\theta:\F[x]\rightarrow \F[x]$, such  that $\theta(x)=x$. This map factors to a map from $\mS'$ to itself,
since $\theta(x^\ell-1)=x^\ell-1$.

The identification is then given by the map $\nu : \F^n  \longrightarrow  \mR',$  defined as
\begin{equation}\label{eq:nu}  \nu((c^{(0)} \mid \ldots \mid c^{(\ell-1)}))= \sum_{j=0}^{N-1} \bigg(\sum_{i=1}^\ell c_j^{(i)}x^i \bigg)z^j.
\end{equation}

The definitions of the maps $\mu$ and $\nu$ allow to give the following characterization of cyclic-skew-cyclic codes.

\begin{theorem}\cite[Theorem 1]{martinez2020sum}
 Let $\C\subseteq \F^n$ be a code. The following are equivalent.
 \begin{enumerate}
     \item $\C$ is a cyclic-skew-cyclic code.
     \item $\mu(\C)$ is a left ideal of $\mR$.
     \item $\nu(\C)$ is a left ideal of $\mR'$.
 \end{enumerate}
\end{theorem}

\begin{remark}\label{rem:PIR}
Note that, $\mR$ and $\mR^\prime$ are (canonically) isomorphic and moreover, whenever the characteristic of $\F$ does not divide $\ell$, $\mR^\prime$ (and hence $\mR$) are principal ideal rings; see \cite{martinez2020sum}. This implies that a cyclic-skew-cyclic code has a generator skew polynomial in $\mS^\prime[z;\theta]$.
\end{remark}

\subsection{Defining set and BCH bound}
In this subsection we define the evaluation map on the $\ell$-th roots of unity in order to introduce the defining set of skew-cyclic-codes and the sum-rank BCH bound. Also for this part, the definitions and the results specialised to the case of finite fields were provided in \cite{martinez2020sum}.

Assume $N=jm$, for some  positive integer $j$ and let $a\in \K$ be an $\ell$-th root of unity. Define the following ring morphism:
\begin{align}\label{eq:skeweval}
 \Ev_{a,z} : \faktor{\left(\faktor{\LL[x]}{(x^\ell - 1)}\right)[z;\sigma]}{(z^N-1)} \longrightarrow \faktor{\left(\faktor{\LL[x]}{(x-a)}\right)[z;\sigma]}{(z^N-1)},   
\end{align}
such that $\Ev_{a,z}(f_0(x)+ \dots + f_{N-1}(x)z^{N-1}) = f_0(a) + \dots + f_{N-1}(a)z^{N-1}$. It is easy to verify that $\Ev_{a,z}$ is a ring homomorphism, that is 
\begin{equation}\label{eq:evaluation_ring_homo}
    \Ev_{a,z}(f(x,z)g(x,z))=\Ev_{a,z}(f(x,z))\Ev_{a,z}(g(x,z)), 
\end{equation}
due to the fact that $a\in\K$ is fixed by $\sigma$; see also \cite{martinez2020sum}.
Moreover, note that 
$$\faktor{\left(\faktor{\LL[x]}{(x-a)}\right)[z;\sigma]}{(z^N-1)}\cong \faktor{\LL[z;\sigma]}{(z^N-1)}.$$

Let $\beta\in \LL^\ast$ and define the following evaluation map:
\begin{align}\label{eq:linearizedeval}
    \Ev_{\beta}^\sigma: \faktor{\LL[z;\sigma]}{(z^N-1)}\longrightarrow \faktor{\LL[z;\sigma]}{(z-\sigma(\beta)\beta^{-1})},
\end{align}
where $\Ev_{\beta}^\sigma(g(z)) = g(\sigma(\beta)\beta^{-1})$, in the sense of \eqref{eq:righteval}.

Finally, by noting that $\faktor{\LL[z;\sigma]}{(z-\sigma(\beta)\beta^{-1})} \cong \LL$, we can define the composition of the maps defined in \eqref{eq:skeweval} and \eqref{eq:linearizedeval} as follows.

\begin{definition}\label{def:totaleval}
Let $a\in \K$ be an $\ell$-th root of unity and $\beta \in \LL^\ast$. Define the \emph{total evaluation map}
\begin{align}\label{eq:toteval}
    \Ev_{a,\beta} :  \faktor{\left(\faktor{\LL[x]}{(x^\ell - 1)}\right)[z;\sigma]}{(z^N-1)} \longrightarrow \faktor{\left(\faktor{\LL[x]}{(x-a)}\right)[z;\sigma]}{(z-\sigma(\beta)\beta^{-1})}
\end{align}
as the composition map $\Ev_\beta^\sigma \circ \Ev_{a,z}.$
\end{definition}

\begin{remark}
 To a skew polynomial $f(z)=\sum_{i=0}^r f_i z^i\in \LL[z;\sigma]$ we can associate a $\sigma$-polynomial $f^\sigma(\sigma) = \sum_{i=0}^r f_i\sigma^i \in \LL[\sigma]$. Here, the evaluation in $\beta \in \LL$ is given by 
 $$f^\sigma(\beta):=\sum_{i=0}^r f_i\sigma^i(\beta).$$
 With this in mind, it holds that $\Ev_{\beta}^\sigma(f)=f(\sigma(\beta)\beta^{-1}) = f^\sigma(\beta)\beta^{-1}$. In other words, the evaluation of a skew polynomial can be related to the usual evaluation of its associated $\sigma$-polynomial.
 \end{remark}
 
 \begin{example}
  In the case of finite fields, when $\LL=\Fm$ and $\sigma(\beta)=\beta^q$ is the $q$-Frobenius automorphism, we have that the $\sigma$-polynomial $f^\sigma$ is the classical linearized polynomial $f^\sigma(y) = \sum_{i=0}^r f_iy^{q^i}$.
 \end{example}
 
In the following, we give the definition of the defining set for a cyclic-skew-cyclic code, as in \cite{martinez2020sum}. As for cyclic codes, it is defined as the set of the zeros of the  generator polynomial (in this case a skew polynomial) which define the code.

\begin{definition}\label{def:defset}
Let $\C\subseteq\mR^\prime$ be a cyclic-skew-cyclic code with generator skew polynomial $g(x,z)\in \mS^\prime[z;\theta]$. Then, the \emph{defining set} $T_\C$ of $\C$ is the set
$$T_\C :=\{(a,\beta) \in\K\times \LL^\ast \mid a^\ell = 1, \Ev_{a,\beta}(g(x,z)) = 0\}. $$

\end{definition}

Finally, we recall the sum-rank BCH bound for cyclic-skew-cyclic codes obtained in \cite{martinez2020sum}. For this result we require that $m=N$.

\begin{theorem}[Sum-rank BCH bound \textnormal{\cite[Theorem 7]{martinez2020sum}}]
Let $b,\delta, r,t$ be integers, such that $\gcd(n,t)=1$. Let $a\in\K$ be a primitive $\ell$-th root of unity and $\beta$ be a normal element of $\LL/\K$. Let $n=m\ell$ and $\C\subseteq\F^n$ be a cyclic-skew-cyclic code. 
If 
$$
\{(a^{b+it}, \sigma^{it}(\beta))\in \K \times \LL^\ast \, : \,  0\leq i\leq \delta-2\}\subseteq T_\C,
$$
 then $\dd_{\SR}(\C)\geq \delta$.
\end{theorem}

\section{Isometries and Equivalence of Sum-Rank Metric Codes}\label{sec:isometries}

In this section we describe the group of semilinear isometries for the sum-rank metric. This concept is known for the Hamming metric and also for the rank-metric and it is fundamental for describing equivalence of codes. 

In order to be more consistent with the usual notation for permutations, in contrast with the rest of the paper in this section we use the vectors indices starting from $1$ instead of $0$. Also, in this section we consider a general partition  $\mathbf{n} = (n_1,\dots,n_l)$ of the integer $n$.
\subsection{Equivalence of Hamming metric codes}
    Here we recall the notion of equivalence of codes in the Hamming metric and their automorphisms groups. Fix the field $\F$ and consider the metric space $(\F^\ell,\dd_{\HH})$. It is well-known that the semilinear isometries of this space form a group given by $(\F^*)^\ell \rtimes (\Aut(\F)\times \mathcal S_\ell)$, where $\Aut(\F)$ denotes the group of field automorphisms of $\F$ and $\mathcal S_\ell$ the symmetric group on $\ell$ elements. The action of this group on $\F^\ell$ is described here:
    $$\begin{array}{rcl}
    ((\F^*)^\ell \rtimes (\Aut(\F)\times \mathcal S_\ell)) \times \F^\ell & \longrightarrow & \F^\ell \\
    ((\mathbf{a},\theta,\pi),v) &\ \longmapsto & (a_{1}\theta(v_{\pi^{-1}(1)}),\ldots, a_{\ell}\theta(v_{\pi^{-1}(\ell)})).
    \end{array}$$
    Such an action can be naturally extended to an action on the set $\mathcal L(\F^\ell)=\{ \C \subseteq \F^\ell: \C \mbox{ is } \F\mbox{-linear }\}$ of $\F$-subspaces of $\F^\ell$, denoted by
    $$ \Psi_{\HH}: ((\F^*)^\ell \rtimes (\Aut(\F)\times \mathcal S_\ell)) \times \mathcal L(\F^\ell)  \longrightarrow  \mathcal L(\F^\ell).$$
    
    \begin{definition}
     Two $[\ell,k]_\F$ codes are called \emph{semilinearly (Hamming) equivalent} if they belong to the same orbit of the action $\Psi_{\HH}$.
    \end{definition}
    
    Notice that in the literature it is more often considered  the notion of \emph{monomial equivalence}. This corresponds to the orbits of the action induced by the subgroup $(\F^*)^\ell\rtimes \mathcal S_\ell$, which only considers the linear isometries. Furthermore, the notion of \emph{permutation equivalence} is obtained by only taking the action of the symmetric group $\mathcal S_\ell$.

    \begin{definition}
     The \emph{automorphism (Hamming) group} of an $[\ell,k]_\F$ code $\C$ endowed with the Hamming metric is the stabilizer of $\C$ in $((\F^*)^\ell \rtimes (\Aut(\F)\times \mathcal S_\ell))$ under the action of $\Psi_{\HH}$, and we denote it by $\Aut_{\HH}(\C)$.
    \end{definition}
    
    \begin{remark}\label{rem:cyclicity}
    Let $\bar{\pi}$ be the $\ell$-cycle $(1\ldots \ell)$. With this notion, we can see that a code $\C$ is cyclic if and only if $(\mathbf{1},\mathrm{id},\bar{\pi})\in \Aut_{\HH}(\C)$. However, when studying properties of a code which are invariant under code equivalence,  it is more convenient to consider when a code $\C$ is \emph{equivalent to a cyclic code}, that is, when in the orbit of $\C$ there exists a cyclic code. This corresponds to the property that there exists an element $\mathfrak g\in ((\F^*)^\ell \rtimes (\Aut(\F)\times \mathcal S_\ell))$ such that $\mathfrak g^{-1}(\mathbf{1},\mathrm{id},\bar{\pi})\mathfrak g\in\Aut_{\HH}(\C)$. Moreover, it is clear that the subgroup $\Aut(\F)$ leaves invariant the property of a code being cyclic. Hence, we can assume that the element $\mathfrak g$ is of the form $\mathfrak g=(\mathbf{a},\mathrm{id},\tau)$.
    
    \noindent In conclusion, since the conjugacy class of an $\ell$-cycle in $\mathcal S_\ell$ is the set of all $\ell$-cycles, we have that  $\C$ is (equivalent to) a cyclic code if and only if there exists a $\mathbf{b}\in (\F^*)^\ell$, and an $\ell$-cycle $\pi$ such that $(\mathbf{b},\mathrm{id},{\pi})\in\Aut_{\HH}(\C)$.
    \end{remark}

 \subsection{Equivalence of rank-metric codes}    
    Now, we describe the notion of isometries and equivalence also for codes in the rank metric. Consider the metric space $(\F^N,\dd_{\RR})$, where $\dd_{\RR}$ denotes the rank metric over $\F/\E$. The semilinear isometries of this metric space are given by the group $(\F^*\times \GL(N,\E))\rtimes \Aut(\F)$, which is acting on $\F^N$ via
        $$\begin{array}{rcl}
   ((\F^*\times \GL(N,\E))\rtimes \Aut(\F)) \times \F^N & \longrightarrow & \F^N \\
    ((\alpha,M,\theta),v) &\ \longmapsto & (\theta(\alpha v_1),\ldots,\theta(\alpha v_N))M.
    \end{array}$$
       Also in this case we extend this action to an action on the set $\mathcal L(\F^N)=\{ \C \subseteq \F^N: \C \mbox{ is } \F\mbox{-linear }\}$ of $\F$-subspaces of $\F^N$, denoted by
    \begin{equation}\label{eq:PsisR} \Psi_{\RR}:  ((\F^*\times \GL(N,\E))\rtimes \Aut(\F)) \times \mathcal L(\F^N)  \longrightarrow  \mathcal L(\F^N).\end{equation}
    
        \begin{definition}
     Two $[N,k]_\F$ codes are called \emph{semilinearly (rank) equivalent} if they belong to the same orbit of the action $\Psi_{\RR}$.
    \end{definition}

    \begin{definition}
     The \emph{automorphism (rank) group} of an $[N,k]_\F$ code $\C$ endowed with the rank metric over $\F/\E$ is the stabilizer of $\C$ in $((\F^*\times \GL(N,\E))\rtimes \Aut(\F))$ under the action of $\Psi_{\RR}$, and it will be denoted by $\Aut_{\RR}(\C)$.
    \end{definition}
    
    \begin{remark}
    First, notice that if one restricts to  $\F$-linear codes, then the action of $\F^*$ is trivial. Hence, in this case we can only consider rank isometries of the form $(1,M,\theta)$. 
     Moreover, in this setting a code $\C$ is formally \emph{skew-cyclic with respect to $\theta$} if $(1,P_{\bar{\pi}},\theta)\in\Aut_{\RR}(\C)$, where $P_{\pi}$ is the permutation matrix acting as $\bar{\pi}$. 
     Depending now on the interest in properties of the code which are invariant under Hamming equivalence or rank equivalence, one may prefer to study codes which are rank equivalent or Hamming equivalent to a   skew-cyclic code with respect to $\theta$. For us, the natural metric inherited by a skew-cyclic code is the rank metric, and hence we will consider a code $\C$ to be (equivalent to) a skew-cyclic code if there exists an element of the form  $(1,M^{-1}P_{\pi}M,\theta)$ in $\Aut_{\RR}(\C)$.
    \end{remark}
    
    \subsection{Equivalence of sum-rank-metric codes}
    
    In this subsection we characterize the linear and semilinear isometries of a space endowed with the sum-rank metric. More precisely, let $\mathbf{n}:=(n_1,\ldots,n_\ell)$ be any partition of $n$, and let $\dd_{\SR}^{\mathbf{n}}$ be the sum-rank metric on the space $\F^n$ with respect to the partition $\mathbf{n}$ and the field extension $\F/\E$. We are going to characterize the $\F$-linear and the semilinear isometries of $(\F^{n},\dd_{\SR}^{\mathbf{n}})$.

{First, we introduce the following notation: for a given partition $\mathbf{n}=(n_1,\ldots,n_{\ell})$ of $n$, we associate the vector of positive integers} $\lambda(\mathbf{n})$, which counts the occurrences of the distinct entries of $\mathbf{n}$. Formally, let $\mathcal N(\mathbf{n}):=\{n_1,\ldots, n_{\ell}\}$, and let $t:=|\mathcal N(\mathbf{n})|$. Let $n_{i_1},\ldots, n_{i_t}$ be the distinct elements of $\mathcal N(\mathbf{n})$ and set 
    $\lambda(\mathbf{n}) \in \mathbb N^{t}$ to be the vector whose entries are
    $$ \lambda_j:=|\{k : n_k=n_{i_j}\}|, \quad \mbox{ for each } j=1,\ldots, t.$$
Furthermore, we denote by $l(\mathbf{n})=\ell$ the length of the partition $\mathbf{n}$.
    
    Now, we fix a vector of positive integers ${\mathbf{v}}=(v_1,\ldots,v_r) \in \mathbb{N}^r$, we denote by $\mathcal S_{\mathbf{v}}$ the direct product of the symmetric groups on $v_i$ elements, for each $i$, that is 
    $$ \mathcal S_{\mathbf{v}}=\mathcal S_{v_1}\times \ldots \times \mathcal S_{v_r}.$$
    Observe that $\mathcal S_{\mathbf{v}}$ is a subgroup of the symmetric group $\mathcal S_v$, where $v=v_1+\ldots+v_r$.
    Furthermore, we denote by $\GL(\mathbf{v},\E)$ the direct product of the general linear groups of degree $v_i$ over the field $\E$, that is
    $$ \GL(\mathbf{v},\E)=\GL(v_1,\E)\times \ldots \times \GL(v_r,\E).$$
    
    With this notation, we are now ready for the main result. We determine the group of $\F$-linear isometries of the space $(\F^n,\dd_{\SR}^{\bfn})$, which we denote by
    $\LIS_{\SR}(\mathbf{n},\F/\E).$
    Notice that the special case $\mathbf{n}=(N,\ldots,N)$ was already shown in \cite[Theorem 2]{martinez2020hamming} (see also \cite[Proposition 4.25]{neri2021twisted} for the $\E$-linear sum-rank isometries). Here we deal with the more general case, even though the strategy of the proof is essentially the same.

    \begin{theorem}\label{thm:linear_isometries}
     Let $\mathbf{n}$ be a partition of $n$. The group of $\F$-linear isometries of $(\F^{n},\dd_{\SR}^{\mathbf{n}})$ is
     $$ ((\F^*)^{l(\mathbf{n})} \times \GL(\mathbf{n},\E)) \rtimes \mathcal S_{\lambda(\mathbf{n})},$$
     which is acting as
     \begin{equation}\label{eq:isometry}(\mathbf{a},M_1,\ldots, M_{l(\mathbf{n})},\pi)\cdot (c^{(1)} \mid \ldots \mid c^{(l(\mathbf{n}))})\longmapsto (a_1c^{(\pi^{-1}(1))}M_1 \mid \ldots \mid a_{l(\mathbf{n})}c^{(\pi^{-1}(l(\mathbf{n})))}M_{l(\mathbf{n})}).\end{equation}
    \end{theorem}
    
    \begin{proof}
     It is immediate to observe that each element of the group $((\F^*)^{l(\mathbf{n})} \times \GL(\mathbf{n},\E)) \rtimes \mathcal S_{\lambda(\mathbf{n})}$ describes an $\F$-linear isometry of $(\F^n,\dd_{\SR}^{\mathbf{n}})$ via \eqref{eq:isometry}. Therefore, $\LIS_{\SR}(\mathbf{n},\F/\E) \supseteq ((\F^*)^{l(\mathbf{n})} \times \GL(\mathbf{n},\E)) \rtimes \mathcal S_{\lambda(\mathbf{n})}$. 
     
     On the other hand, let $f\in \LIS_{\SR}(\mathbf{n},\F/\E)$. For each $i \in \{1,\ldots, n\}$, denote by $e_i \in \F^n$ the $i$th standard basis vector. Moreover, we write $I_1,\ldots, I_{l(\mathbf{n})}$ to refer to the sets of indices of each of the  coordinates blocks, where $|I_j|=n_j$ for each $j$. Since $f$ is an isometry, we have $\wt_{\SR}(f(e_i))=\wt_{\SR}(e_i)=1$. Consider the vector $e_1$. By definition of sum-rank metric, this means that $v_1:=f(e_1)$ is zero on all but one coordinate blocks, i.e. there exists  $j$ such that $v_1=(v^{(1)} \mid \cdots \mid v^{(l(\mathbf{n}))})$ with $\wt_{\RR}(v^{(j)})=1$ and $v^{(i)}=0$ for each $i \neq j$. Select now any other index $t \in I_1$, and again, for the same reason, there exists $j_t$ such that for $f(e_t)=:v_t=(v^{(1)} \mid \cdots \mid v^{(l(\mathbf{n}))})$ it holds $\wt_{\RR}(v^{(j_t)})=1$ and $v^{(i)}=0$ for each $i \neq j_t$. Now observe that $1,t \in I_1$ and therefore we have $\wt_{\SR}(e_1+e_t)=\wt_{\RR}(e_1^{(1)}+e_t^{(1)})=1$. Hence, we must have that the nonzero coordinate blocks of $f(e_1)$ and $f(e_t)$ are the same, that is $j_t=j$. Since this holds for each $t \in I_1$ and we can do the same reasoning starting from each coordinate block $I_j$, this shows that $f$ induces a permutation $\pi_f$ between the coordinate blocks.  Therefore, $f$ also induces a map $f_j:\F^{n_j}\rightarrow \F^{n_{\pi_f(j)}}$ for each $j$. Since $f$ is bijective, also $f_j$ must be bijective, and therefore $|I_{\pi_f(j)}|=n_{\pi_f(j)}=n_j=|I_j|$. Thus, the permutation $\pi_f$ permutes the coordinate blocks, by sending each block in another with the same cardinality. This implies that $\pi_f \in \mathcal S_{\lambda(\mathbf{n})}$. Finally, by definition of sum-rank metric, each map $f_j$ has to be an $\F$-linear isometry of $(\F_{n_j},\dd_{\RR})$. These isometries are known to form the group $(\F^* \times \GL(n_j,\E))$, where $f_j$ acts as $:v\longmapsto a_jvM_j$, for some $a_j \in \F^*$, $M_j \in \GL(n_j,\E)$, as described by the linear part of the map $\Psi_{\RR}$ in \eqref{eq:PsiR}; see  \cite[Theorem 1]{berger2003isometries}, \cite[Proposition 1]{morrison2014equivalence}. This shows that the $\F$-linear isometry $f$ acts as the element $(\mathbf{a},M_1,\ldots, M_{l(\mathbf{n})},\pi_f)$, and therefore $\LIS_{\SR}(\mathbf{n},\F/\E)\subseteq ((\F^*)^{l(\mathbf{n})} \times \GL(\mathbf{n},\E)) \rtimes \mathcal S_{\lambda(\mathbf{n})}$, which conludes the proof.
    \end{proof}
    
    Theorem \ref{thm:linear_isometries} immediately leads to the following corollary, which determine the group of semilinear isometries of $(\F^n,\dd_{\SR}^{\mathbf{n}})$.
    
        \begin{corollary}
     Let $\mathbf{n}$ be a partition of $n$. The group of semilinear isometries of $(\F^{n},\dd_{\SR}^{\mathbf{n}})$ is
     $$ (((\F^*)^{l(\mathbf{n})} \times \GL(\mathbf{n},\E)) \rtimes \mathcal S_{\lambda(\mathbf{n})})\rtimes \Aut(\F),$$
     which is acting as
     \begin{equation}\small{\label{eq:sl_isometry}(\mathbf{a},M_1,\ldots, M_{l(\mathbf{n})},\pi,\theta)\cdot (c^{(1)} \mid \ldots \mid c^{(l(\mathbf{n}))})\longmapsto (\theta(a_1c^{(\pi^{-1}(1))})M_1 \mid \ldots \mid \theta(a_{l(\mathbf{n})}c^{(\pi^{-1}(l(\mathbf{n})))})M_{l(\mathbf{n})}).}\end{equation}
    \end{corollary}

    It is well-known that already for the case $\ell=1$ and $n_1>1$, MacWilliams extension theorem does not hold (see e.g. \cite[Example 2.9]{barra2015macwilliams}). Thus, we give a global definition of equivalence of codes in the sum-rank metric, i.e. we only consider when they are related by an isometry of the whole ambient space. Define the map
    \begin{equation}\label{eq:PsiR} \Psi_{\SR}: ((((\F^*)^{l(\mathbf{n})} \times \GL(\mathbf{n},\E)) \rtimes \mathcal S_{\lambda(\mathbf{n})})\rtimes \Aut(\F)) \times \mathcal L(\F^n)  \longrightarrow  \mathcal L(\F^n).\end{equation}
   as the extension of the map acting as in \eqref{eq:sl_isometry} to the set $\mathcal L(\F^n)=\{\C \subseteq \F^n : \C \mbox{ is } \F\mbox{-linear}\}$ of $\F$-subspaces.

    \begin{definition}
         Two $[n,k]_\F$ codes are called \emph{semilinearly (sum-rank) equivalent} if they belong to the same orbit of the action $\Psi_{\SR}$.
    \end{definition}

    \begin{definition}
     The \emph{automorphism (sum-rank) group} of an $[n,k]_\F$ code $\C$ endowed with the sum-rank metric $\dd_{\SR}^{\mathbf{n}}$ over $\F/\E$ is the stabilizer of $\C$ in $(((\F^*)^{l(\mathbf{n})} \times \GL(\mathbf{n},\E)) \rtimes \mathcal S_{\lambda(\mathbf{n})})\rtimes \Aut(\F) $ under the action of $\Psi_{\SR}$, and it will be denoted by $\Aut_{\SR}(\C)$.
    \end{definition}
    
\begin{remark}
 Fix the partition $\mathbf{n}=(N,\ldots,N)$. Then a code is cyclic-skew-cyclic with respect to the automorphism $\theta$ if and only if $(\mathbf{1},(I_N,\ldots,I_N),\bar{\tau},\mathrm{id})$ and $(\mathbf{1},(P_{\bar{\pi}},\ldots,P_{\bar{\pi}}),\mathrm{id},\theta)$ belong both to $\Aut_{\SR}(\C)$, where $\bar{\tau}=(1\ldots\ell)$ and $\bar{\pi}=(1\ldots n)$.
\end{remark}

\section{Roos and Hartmann-Tzeng bounds for cyclic-skew-cyclic codes}\label{sec:Roos_HT_bounds}

In this section we provide a lower bound on the minimum sum-rank distance of cyclic-skew-cyclic code which generalizes the Roos bound for cyclic codes in the Hamming metric \cite{roos1982generalization, roos1983new} and its skew-cyclic version for the rank metric \cite{alfarano2021roos}. Moreover, we also derive a sum-rank metric version of the Hartmann-Tzeng bound, generalizing the one known for the Hamming metric \cite{hartmann1972generalizations} and the more recent one for the rank metric \cite{Gomez/Lobillo/Navarro/Neri:2018}. 
From now on we fix the partition of $n$ to be $\mathbf{n}=(\underbrace{m,\dots, m}_{\ell \textnormal{ times}})$, hence, we require $N=m$, as for the sum-rank BCH bound.

   Let $a\in \K$ be an $\ell$-th root of unity and $\beta\in\LL^\ast$ be a normal element of $\LL/\K$ and let $\mathcal{A}:=\{1,a,\dots,a^{\ell-1}\}$. Since $\beta a^{bi}$ is still normal in $\LL$ for every $b\geq 0$ and every $0\leq i \leq \ell-1$, we can define new bases of $\LL/\mathbb{K}$ as  $$\Tilde{\mathcal{B}}_i:=\{\beta a^{bi}, \sigma(\beta)a^{bi}, \dots, \sigma^{m-1}(\beta)a^{bi}\}.$$ 
   Let $\Tilde{\mathcal{B}}=(\Tilde{\mathcal{B}_0}, \dots, \Tilde{\mathcal{B}}_{\ell-1})$ and for $k\in\{1,\dots, n\}$, define the matrix $$D(\mathcal{A},\Tilde{\mathcal{B}}) = (D_0|\dots| D_{\ell-1})\in\LL^{k\times n},$$
   where $n=\ell m$ and 
   \begin{equation}\label{eq:DiRS} 
   D_i = \begin{pmatrix}
   \beta a^{bi} & \sigma(\beta) a^{bi} & \cdots & \sigma^{m-1}(\beta)a^{bi} \\
   \sigma(\beta) a^{(b+1)i} & \sigma^2(\beta) a^{(b+1)i} &\cdots & \beta a^{(b+1)i} \\
   \vdots & \vdots & \ddots & \vdots \\
   \sigma^{k-1}(\beta) a^{(b+k-1)i} &  \sigma^{k}(\beta) a^{(b+k-1)i} & \cdots &  \sigma^{k-2}(\beta) a^{(b+k-1)i}
   \end{pmatrix}.\end{equation}
   Then, $D(\mathcal{A},\Tilde{\mathcal{B}})$ is the generator matrix of a linearized Reed-Solomon code $\C_{k}^\sigma(\mathcal{A},\Tilde{\mathcal{B}})$; see \cite{martinez2018skew}.

   Let $\mathcal{E}= \{\beta_1,\dots,\beta_m\}$ be a basis of $\LL/\K$, assume that $\ell$ is coprime with the characteristic of $\K$ and with $m$ and let $a\in\K$ be an $\ell$-th root of unity. Fix $b\geq 0$. Define $\mB_0,\dots, \mB_{\ell-1}$ bases of the extension $\LL/\K$ as 
   $$ \mB_i := \mathcal{E}\cdot a^{bi} =  \{\beta_1a^{bi}, \dots, \beta_ma^{bi} \}.$$
   \begin{lemma}\label{lem:rankDi-basecase}
   Let $t,r$ be positive integers and $k_0,\dots,k_r\in\{0,\dots, n-1\}$ be such that $k_r - k_0 \leq t + r - 1$ and $k_{j-1} < k_j$ for $1 \leq j \leq r$. For any $0\leq i\leq \ell-1$ let $\alpha_1^{(i)},\dots,\alpha_{j_i}^{(i)}\in\mB_i$ such that $j_0+ \dots + j_{\ell-1} = t+r$. Then the matrix 
   $$ A_0 = (\tilde{D}_0|\cdots | \tilde{D}_{\ell-1}), $$
   
   where 
  \begin{equation}\label{eq:tildeDi}
      \tilde{D}_i = \begin{pmatrix}
      \sigma^{k_0}(\alpha_1^{(i)})a^{k_0i} & \cdots & \sigma^{k_0}(\alpha_{j_i}^{(i)})a^{k_0i} \\
      \sigma^{k_1}(\alpha_1^{(i)})a^{k_1i} & \cdots &  \sigma^{k_1}(\alpha_{j_i}^{(i)})a^{k_1i}\\
      \vdots  &  & \vdots \\
      \sigma^{k_r}(\alpha_1^{(i)})a^{k_ri} & \cdots &   \sigma^{k_r}(\alpha_{j_i}^{(i)})a^{k_ri}
  \end{pmatrix}
  \end{equation}
   has rank $r+1$.
   \end{lemma}
  \begin{proof}
  Define $\bar{\mB}:=(\sigma^{k_0}(\mathcal{B}_0)\cdot a^{k_0i},\ldots, \sigma^{k_0}(\mathcal{B}_{\ell-1})\cdot a^{k_0i})$ where $\sigma^{k_0}$ is applied to each element of the bases $\mB_i$'s. Let $\mathcal{A} = \{1,a,\dots, a^{\ell-1}\}$.
  Consider the linearized Reed-Solomon $\C_{r+t}^\sigma(\mathcal{A},\bar{\mB})$ with generator matrix $D(\mathcal{A},\bar{\mB}) = (D_0|\dots|D_{\ell-1})$, where $D_i$ is defined as in Equation \eqref{eq:DiRS}. Then, the matrix $A_0$ is obtained from $D(\mathcal{A},\bar{\mB})$ after deleting $(t-1)$ rows of a full-size submatrix $(t+r)\times (t+r)$. Since $\C_{r+t}^\sigma(\mathcal{A},\bar{\mB})$ is an MDS code, it follows that the rank of the submatrix $A_0$ is full, which concludes the proof.
  \end{proof}
  
  Before proceeding with the next auxiliary result, we introduce a useful map as follows. For a  root of unity $b\in \K$, for $\varsigma \in \Gal(\LL/\K)$  and  integers $t_1,t_2$, define 
  \begin{align*}
      \psi_{b,\varsigma}^{(t_1,t_2)} : \LL^{\ell m}  &\longrightarrow  \LL^{\ell m} \\
      (c^{(0)} \mid c^{(1)} \mid \ldots \mid c^{(\ell-1)}) &\longmapsto (\varsigma^{t_1}(c^{(0)}) \mid \varsigma^{t_1}(c^{(1)})b^{t_2} \mid \ldots \mid \varsigma^{t_1}(c^{(\ell-1)})b^{(\ell-1)t_2}).
  \end{align*}
Since $b$ is fixed by $\varsigma$, one can immediately see that the map $\psi_{b,\varsigma}^{(t_1,t_2)}$ is $\LL$-linear and that $\ker(\psi_{b,\varsigma}^{(t_1,t_2)})=\{0\}$. Moreover, it can be readily observed that \begin{equation} \label{eq:composition}
    \psi_{b,\varsigma}^{(t_1,t_2)} \circ \psi_{b,\varsigma}^{(u_1,u_2)}=\psi_{b,\varsigma}^{(t_1+u_1,t_2+u_2)},
\end{equation}
where $t_1+u_1$ is taken modulo $\ord_{\K^*}(b)$, and $t_2+u_2$ is taken modulo $\ord_{\Aut(\LL)}(\varsigma)$.

   \begin{lemma}\label{lem:rankDi}
   With the same assumptions of Lemma \ref{lem:rankDi-basecase}, let $s$ be a positive integer coprime with $\ell$ and $m$. Let $A_0 = (\tilde{D}_0|\dots|\tilde{D}_{\ell-1})$, where $\tilde{D}_i$'s are defined in \eqref{eq:tildeDi}, and let
   $$A_i:= \left(\begin{array}{c|c|c|c}
   \tilde{D}_0 & \tilde{D}_1 & \cdots &\tilde{D}_{\ell-1} \\
   \sigma^s(\tilde{D}_0)& \sigma^s(\tilde{D}_1)a^s & \cdots & \sigma^s(\tilde{D}_{\ell-1})a^{(\ell-1)s}\\
   \vdots & \vdots& & \vdots  \\
   \sigma^{is}(\tilde{D}_0) &\sigma^{is}(\tilde{D}_1)a^{is} & \cdots & \sigma^{is}(\tilde{D}_{\ell-1})a^{i(\ell-1)s}
   \end{array}\right)$$
   for $i\leq t-1$. Then $\rk(A_i)\geq r+i+1$. In particular, $\rk(A_{t-1}) = r+t$.
   \end{lemma}
\begin{proof}
First observe that the claim holds for  $i=0$, since  by Lemma \ref{lem:rankDi-basecase}, $\rk(A_0) = r+1$. 

\noindent 
Define $\mU_i:=\rowsp(A_i)$, and observe that they form a chain $\mU_0\subseteq \mU_1\subseteq \ldots \subseteq \mU_{t-1}$. Suppose that there exists an $i\geq 1$ such that $\dim(\mU_{i-1})\geq r+i$, but $\dim(\mU_{i})<r+i+1$. This implies that $\dim(\mU_{i-1})=\dim(\mU_{i})=r+i$ and hence $\mU_{i-1}=\mU_i$. However, due to the structure of the matrices $A_j$'s, we have $\mU_i=\mU_{i-1}+\psi_{a,\sigma}^{(s,s)}(\mU_{i-1})$. Therefore, $\psi_{a,\sigma}^{(s,s)}(\mU_{i-1})=\mU_{i-1}$. Let $v$ be the inverse of $s$ modulo $n=\ell m$. Iterating $v$ times the map $\psi_{a,\sigma}^{(s,s)}$ and using \eqref{eq:composition}, we obtain
$$ \mU_i=\mU_{i-1}=(\psi_{a,\sigma}^{(s,s)})^{v}(\mU_{i-1})=\psi_{a,\sigma}^{(1,1)}(\mU_{i-1}).$$
From this, we also obtain that $\psi_{a,\sigma}^{(j,j)}(\mU_{i-1})=\mU_{i-1}$, for every $j\in \{1,\ldots, i+r-1\}$. Hence, $\mU_{i}$ contains the row space of the matrix

$$ \left(\begin{array}{c|c|c|c}
   \tilde{D}_0 & \tilde{D}_1 & \cdots &\tilde{D}_{\ell-1} \\
   \sigma(\tilde{D}_0)& \sigma(\tilde{D}_1)a & \cdots & \sigma(\tilde{D}_{\ell-1})a^{(\ell-1)}\\
   \vdots & \vdots& & \vdots  \\
   \sigma^{i+r-1}(\tilde{D}_0) &\sigma^{i+r-1}(\tilde{D}_1)a^{i+r-1} & \cdots & \sigma^{i+r-1}(\tilde{D}_{\ell-1})a^{(i+r-1)(\ell-1)}
   \end{array}\right). $$
   From the above matrix, one can select the first row from each block and obtain the submatrix
   $$  X = (\tilde{E}_0|\cdots | \tilde{E}_{\ell-1}), $$
   where 
  \begin{equation*}
      \tilde{E}_i = \begin{pmatrix}
      \sigma^{k_0}(\alpha_1^{(i)})a^{k_0i} & \cdots & \sigma^{k_0}(\alpha_{j_i}^{(i)})a^{k_0i} \\
      \sigma^{k_0+1}(\alpha_1^{(i)})a^{(k_0+1)i} & \cdots &  \sigma^{k_0+1}(\alpha_{j_i}^{(i)})a^{(k_0+1)i}\\
      \vdots  &  & \vdots \\
      \sigma^{k_0+i+r-1}(\alpha_1^{(i)})a^{(k_0+i+r-1)i} & \cdots &   \sigma^{k_0+i+r-1}(\alpha_{j_i}^{(i)})a^{(k_0+i+r-1)i}
  \end{pmatrix}.
  \end{equation*}
 By Lemma \ref{lem:rankDi-basecase}, $X$ has rank $r+i$ and we obtain a contradiction, which concludes the proof.

\end{proof}

We recall the following result which will be used for proving the Roos bound for the sum-rank metric. For this purpose, let $\C$ be an $[n,k]_K$ code and $A \in \GL(n,K)$, then define $\C\cdot A:=\{cA \mid c \in \C\}$.

    \begin{theorem}[\text{\cite[Theorem 3]{martinez2019theory}}]\label{thm:SRminHamm}
     Let $\C$ be an $[n,k]$ sum-rank metric code w.r.t. the extension $\F/\E$ and the partition $(n_1,\ldots, n_\ell)$ of $n$. Then
     $$\dd_{\SR}(\C)=\min\{ \dd_{\HH}(\C\cdot A) \mid A=\mathrm{diag}(A_1,\ldots,A_\ell), A_i \in \GL(n_i,\E)\}.$$

     \end{theorem}

We are now ready to prove the main result of the section which partially answers to Open Problem 2 in \cite{martinez2020sum}.

\begin{theorem}(Sum-rank Roos bound)\label{thm:Roos_SR}
  Let $n=m\ell$ and $\C\subseteq\F^n$ be a cyclic-skew-cyclic code. Let $b, s,\delta, k_0,\dots, k_r$ be integers, such that $\gcd(n,s)=1$, $k_i<k_{i+1}$ for $i=0,\dots,r-1$, $k_r-k_0 \leq \delta + r -2$. Let $a\in\K$ be a primitive $\ell$-th root of unity and $\beta$ be a normal element of $\LL/\K$. 
 If 
$$
\{(a^{b+si+k_j}, \sigma^{si+k_j}(\beta))\in \K \times \LL^\ast \, : \,  0\leq i\leq \delta-2, 0\leq j\leq r\}\subseteq T_\C,
$$
 then $\dd_{\SR}(\C)\geq \delta + r$.
\end{theorem}
\begin{proof}
Let $c(x,z)=\sum_{t=0}^{m-1}f_t(x)z^{t}\in\C$, with $f_t(x) \in\mS^\prime$. We can write $$c(x,z)=\sum_{t=0}^{m-1}\sum_{h=0}^{\ell-1}f_{t,h}x^{h}z^{t},$$
where $f_{t,h}\in\F$. For any $0\leq i\leq \delta-2$ and $0\leq j\leq r$, let $u_{i,j}:=a^{b+si+k_j}$ and $v_{i,j}:=\sigma^{si+k_j}(\beta)$. Now, we apply the total evaluation map  $\Ev_{u_{i,j},v_{i,j}}$ defined in \eqref{eq:toteval} to $c(x,z)$, obtaining
\begin{align*}
     0 =\Ev_{u_{i,j},v_{i,j}}(c(x,z)) &= \sum_{t=0}^{m-1}\left(\sum_{h=0}^{\ell-1}f_{t,h}u_{i,j}^{h}\right)\sigma^{t}(v_{i,j})v_{i,j}^{-1}   \\
     &= \sum_{h=0}^{\ell-1}\left(\sum_{t=0}^{m-1} f_{t,h}\sigma^{t}(v_{i,j}) \right)u_{i,j}^{h}v_{i,j}^{-1}.
\end{align*}
The expression above holds for every $0\leq i\leq \delta-2$ and $0\leq j\leq r$, which means that the codeword $c(x,z)$ is in the left kernel of the matrix
\begin{equation}\label{eq:leftkernel}
    \left(\begin{array}{c|c|c|c}
   {E}_0 & E_1 & \cdots & E_{\ell-1} \\
   \sigma^s(E_0)& \sigma^s(E_1)a^s & \cdots & \sigma^s(E_{\ell-1})a^{s}\\
   \vdots & \vdots& & \vdots  \\
   \sigma^{(\delta-2)s}(E_0) &\sigma^{(\delta-2)s}(E_1)a^{(\delta-2)s} & \cdots & \sigma^{(\delta-2)s}(E_{\ell-1})a^{(\delta-2)s}
   \end{array}\right),
\end{equation} 
 where 
 \begin{equation}\label{eq:kernel}
 E_i =\begin{pmatrix}
 \sigma^{k_0}(\beta) a^{(b+k_0)i} & \sigma^{k_0 + 1}(\beta)a^{(b+k_0)i} &\cdots & \sigma^{k_0+m-1}(\beta)a^{(b+k_0)i} \\
  \sigma^{k_1}(\beta) a^{(b+k_1)i} & \sigma^{k_1 + 1}(\beta)a^{(b+k_1)i} &\cdots & \sigma^{k_1+m-1}(\beta)a^{(b+k_1)i}\\
  \vdots & \vdots & \cdots &\vdots \\
   \sigma^{k_r}(\beta) a^{(b+k_r)i} & \sigma^{k_r + 1}(\beta)a^{(b+k_r)i} &\cdots & \sigma^{k_r+m-1}(\beta)a^{(b+k_r)i}
 \end{pmatrix}.
 \end{equation}
Let $w\leq \delta + r - 1$ and assume there is a codeword $\bar{c}(x,z)\in\C$ of Hamming weight $\wt(\bar{c}(x,z))=w$, which means that $\bar{c}(x,z)=\sum_{j=0}^{w}f_j(x)z^{h_j}$, where  $f_j(x)\in \mS^\prime$ and $h_j\in\{0,\dots, m-1\}$. By following the same reasoning of above, $\bar{c}(x,z)$ is in the left kernel of a matrix obtained from \eqref{eq:kernel} by selecting only the columns corresponding to the nonzero positions of $\bar{c}(x,z)$. By Lemma \ref{lem:rankDi} (with $t=\delta-1$), such a selected submatrix has rank equal to $w$. This implies that $\bar{c}(x,z)=0$. So $\bar{c}(x,z)=0$ is the only codeword having Hamming weight at most $\delta+r-1$, which shows that $\dd_{\HH}(\C) \geq \delta+r$. Now, let $A=\diag(A_0,\dots,A_{\ell-1})$ be any matrix  with $A_i\in\GL(m,\E)$. Observe that by multiplying the matrix in  \eqref{eq:leftkernel} with $A$, we obtain 
\begin{equation*}
    \left(\begin{array}{c|c|c|c}
   {E}_0A_0 & E_1A_1 & \cdots & E_{\ell-1}A_{\ell-1} \\
   \sigma^s(E_0A_0)& \sigma^s(E_1A_1)a^s & \cdots & \sigma^s(E_{\ell-1}A_{\ell-1})a^{s}\\
   \vdots & \vdots& & \vdots  \\
   \sigma^{(\delta-2)s}(E_0A_0) &\sigma^{(\delta-2)s}(E_1A_1)a^{(\delta-2)s} & \cdots & \sigma^{(\delta-2)s}(E_{\ell-1}A_{\ell-1})a^{(\delta-2)s}
   \end{array}\right),
\end{equation*} 
 where 
 \begin{equation*}
 E_iA_i =\begin{pmatrix}
 \sigma^{k_0}(\alpha^{(i)}_1) a^{(b+k_0)i} & \sigma^{k_0}(\alpha^{(i)}_2)a^{(b+k_0)i} &\cdots & \sigma^{k_0}(\alpha^{(i)}_m)a^{(b+k_0)i} \\
  \sigma^{k_1}(\alpha^{(i)}_1) a^{(b+k_1)i} & \sigma^{k_1}(\alpha^{(i)}_2)a^{(b+k_1)i} &\cdots & \sigma^{k_1}(\alpha^{(i)}_m)a^{(b+k_1)i}\\
  \vdots & \vdots & \cdots &\vdots \\
   \sigma^{k_r}(\alpha^{(i)}_1) a^{(b+k_r)i} & \sigma^{k_r}(\alpha^{(i)}_2)a^{(b+k_r)i} &\cdots & \sigma^{k_r}(\alpha^{(i)}_m)a^{(b+k_r)i}
 \end{pmatrix}.
 \end{equation*}
 and $(\alpha^{(i)}_1,\ldots,\alpha^{(i)}_m):=(\beta, \sigma(\beta),\ldots,\sigma^{m-1}(\beta))A_i$.
Thus,  applying the same reasoning as above together with Lemma \ref{lem:rankDi}, we get  $\dd_{\HH}(\C \cdot A)\geq \delta+r$.
We conclude the proof using Theorem~\ref{thm:SRminHamm}.
\end{proof}

With the same strategy of the previous proof, we can show the Hartmann-Tzeng bound for the sum-rank metric, of which we omit the proof.

\begin{theorem}(Sum-Rank HT bound)\label{thm:HT_SR}
 Let $n=m\ell$ and $\C\subseteq\F^n$ be a cyclic-skew-cyclic code.  Let $b,\delta, r,t_1,t_2$ be integers, such that $\gcd(n,t_1)=1$, $\gcd(n,t_2)<\delta$. Let $a\in\K$ be a primitive $\ell$-th root of unity and $\beta$ be a normal element of $\LL/\K$.
If 
$$
\{(a^{b+it_1+st_2}, \sigma^{it_1+st_2}(\beta))\in \K \times \LL^\ast \, : \,  0\leq i\leq \delta-2, 0\leq s\leq r\}\subseteq T_\C,
$$
 then $\dd_{\SR}(\C)\geq \delta + r$.
\end{theorem}

\section{Product Codes}\label{sec:product_codes}
Let $\C_1$ and $\C_2$ be two codes defined over the same field. In this section we study the metric properties inherited by the code $\C_1\otimes \C_2$ according to the metrics which are considered on $\C_1$ and $\C_2$. We show that if we fix  the partition of $n=\ell N$ given by $\mathbf{n}=(N,\ldots,N)$ and we consider  $\C_1$ endowed with the Hamming metric and $\C_2$ with the rank metric, the tensor product $\C_1\otimes \C_2$ is naturally endowed with the sum-rank metric with respect to the partition $\mathbf{n}=(N,\ldots,N)$. We, then, specialize to the case in which we consider the first code to be a cyclic code and the second one to be a skew-cyclic code.

     \subsection{Parameters of product codes}

    Given two linear codes defined over the same field, we recall the notion of \emph{product code}. Consider the vector representation of a tensor product, defined as follows. Let $\ell, N$ be arbitrary positive integers. For $u=(u_0,\ldots,u_{\ell-1}) \in \F^{\ell}, v=(v_0,\ldots,v_{N-1}) \in \F^{N}$, we define $a\otimes b$ to be the vectorization in $\F^{\ell N}$ of their tensor product, that is
    $$ u\otimes v:=(u_0v \mid u_1 v \mid \ldots \mid u_{\ell-1}v).$$
    
    \begin{definition}
     Let $\C_1$ be an $[\ell,k_1]_{\F}$ code and let $\C_2$ be an $[N,k_2]_{\F}$ code. The \emph{product code} between $\C_1$ and $\C_2$ is the $[\ell N,k_1,k_2]_{\F}$ code 
     $$\C_1\otimes \C_2:=\langle \{u\otimes v \mid u \in \C_1, v \in \C_2 \} \rangle_{\F}.$$
    \end{definition}
    
    Equivalently, one can look at elements $u\otimes v$ without performing a vectorization process, simply as matrices in $\F^{\ell \times N}$. In this framework, the product code $\C_1 \otimes \C_2$ consists of all matrices in $\F^{\ell \times N}$ whose columns belong to $\C_1$ and whose rows belong to $\C_2$.
    
    Notice that, the product code $\C_1 \otimes \C_2$ naturally  inherits a partition of its entries as $(N,\ldots, N)$. Moreover, it also induces naturally a metric, which is inherited from the metrics we are equipping the two constituting codes. If we choose the first code $\C_1$ to be an $[\ell, k_1]_{\F}$ code endowed with the Hamming metric, and $\C_2$ to be an $[N, k_2]_{\F}$ code endowed with the rank metric for $\F/\E$, the product code is naturally endowed with the sum-rank metric for $\F/\E$, with respect to the partition $\mathbf{n}=(N,\ldots, N)$. In other words, the sum-rank metric can be seen as the tensor product between the Hamming and the rank metric. This is formally explained in the following two results.
    
    \begin{lemma}\label{lem:sumrank_product_codewords}
     For every  $u \in \F^\ell$, $v \in \F^N$, we have
     $\wt_{\SR}(u\otimes v)=\wt_{\HH}(u)\wt_{\RR}(v)$.
    \end{lemma}
    
    \begin{proof}
    For $u=(u_0,\ldots, u_{\ell-1})$, we write $u\otimes v=(u_0v \mid \ldots \mid u_{\ell-1} v)$. By definition, we have
    $$\wt_{\SR}(u\otimes v)=\sum_{i=0}^{\ell-1} \wt_{\RR}(u_iv)=\sum_{i: u_i \neq 0}\wt_{\RR}(u_iv)=\wt_{\HH}(u)\wt_{\RR}(v).$$
    \end{proof}
    
    \begin{proposition}\label{prop:sumrank_product_code}
     Let $\C_1$ be an $[\ell, k_1]_{\F}$ code and let $\C_2$ be an $[N, k_2]_{\F/\E}$ code. Then 
     $$\dd_{\SR}(\C_1 \otimes \C_2)=\dd_{\HH}(\C_1)\dd_{\RR}(\C_2).$$
    \end{proposition}
    
    \begin{proof}
    Let $t:=\dd_{\HH}(\C_1)$. For ease of exposition, consider a nonzero codeword $c$ in $\C_1\otimes \C_2$ as a matrix in $\F^{\ell\times N}$, where each column belongs to $\C_1$ and each row belongs to $\C_2$. We denote the rows of $c$ as $c^{(1)}, \ldots, c^{(\ell)}$, which is consistent with \eqref{eq:vector_partition}. Since $c$ is nonzero, then there exists at least a nonzero column. Such a column is an element of $\C_1$, and therefore has at least $t$ nonzero entries, say $i_1,\ldots, i_t$. The corresponding rows are therefore nonzero elements of $\C_2$. Hence, 
    $$\wt_{\SR}(c)=\sum_{j=1}^t \wt_{\RR}(c^{(i_j)})\geq \sum_{j=1}^t \dd_{\RR}(\C_2)=t\cdot \dd_{\RR}(\C_2)=\dd_{\HH}(\C_1)\dd_{\RR}(\C_2).$$
    
     Since for any $u \in \C_1$, $v \in \C_2$, the element $u \otimes v$ belongs to $\C_1 \otimes \C_2$, then by Lemma \ref{lem:sumrank_product_codewords} we get the equality.
    \end{proof}

\begin{remark}\label{rem:tensor_metrics}
 Observe that the above discussion and results on tensor products is valid in a more general setting. The proofs of Lemma \ref{lem:sumrank_product_codewords} and of Proposition \ref{prop:sumrank_product_code} are still true if we substitute the rank metric on the second code with any other metric $\dd_{\XX}$, and the sum-rank metric in the ambient space $\F^{\ell N}$ with a metric obtained from $\wt_{\XX}$ by extending it on $\ell$ copies of $\F^N$ using additivity. More precisely, let $\wt_{\XX}:\F^N\rightarrow \mathbb R_{\geq 0}$ be a weight function that induces a metric $\dd_{\XX}:\F^N\times \F^N\rightarrow \mathbb R_{\geq 0}$ given by $\dd_{\XX}(a,b):=\wt_{\XX}(a-b)$. One can define the map
 $$\begin{array}{rccl} 
 \wt_{\SX}:&\F^{\ell N} & \longrightarrow & \mathbb R_{\geq 0}\\
 &(c^{(1)} \mid \ldots \mid c^{(\ell)})& \longmapsto & \sum\limits_{i=1}^{\ell} \wt_{\XX}(c^{(i)}).
 \end{array} $$
 Analogously,  $\dd_{\SR}(a,b):=\wt_{\SR}(a-b)$. With this setting, for every  $u \in \F^\ell$, $v \in \F^N$, we have
     $\wt_{\SR}(u\otimes v)=\wt_{\HH}(u)\wt_{\RR}(v)$, and 
     for every $[\ell, k_1]_{\F}$ code $\C_1$ and $[N, k_2]_{\F}$ code $\C_2$, we have 
     $$\dd_{\SX}(\C_1 \otimes \C_2)=\dd_{\HH}(\C_1)\dd_{\XX}(\C_2).$$
 In other words, \begin{equation}\label{eq:tensor_metrics} (\F^{\ell N},\dd_{\SX})=(\F^\ell,\dd_{\HH})\otimes (\F^N,\dd_{\XX}).\end{equation}
\end{remark}    
    
     \subsection{Algebraic Structure of cyclic-skew-cyclic codes}

    Proposition \ref{prop:sumrank_product_code} shows that the tensor product of a rank-metric code and a Hamming-metric code naturally inherits a structure of a sum-rank metric code, as also explained in Remark \ref{rem:tensor_metrics} in a more general setting; see \eqref{eq:tensor_metrics}. 
    
    However, the metric properties are not the only properties that behave well with the tensoring operation. Indeed, we are going to see  that also the cyclicity structure is somehow preserved. More precisely, we aim to show that the tensor product of a cyclic code and a skew-cyclic code is a cyclic-skew-cyclic code. We first prove the following result on the automorphism group.

    \begin{proposition}\label{prop:tensor_product_autgroup}
  Let $\C_1$ be an $[\ell, k_1]_{\F}$ code and let $\C_2$ be an $[N, k_2]_{\F/\E}$ code. Then
  $$\Aut_{\SR}(\C_1 \otimes \C_2)\supseteq \langle \iota_{\HH}(\Aut_{\HH}(\C_1)), \iota_{\RR}(\Aut_{\RR}(\C_2))\rangle, $$
where, 
    $$\begin{array}{rcl}\iota_{\HH}: ((\F^*)^\ell \rtimes (\Aut(\F)\times \mathcal S_\ell))& \longrightarrow & (((\F^*)^{\ell} \times \GL(\mathbf{n},\E)) \rtimes \mathcal S_{\lambda(\mathbf{n})})\rtimes \Aut(\F)\\
    (\mathbf{a},\theta_1,\pi) & \longmapsto & (\mathbf{a}, \underbrace{I_N, \ldots, I_N}_{\ell \mbox{ \tiny{times} } }, \pi, \theta_1),
    \end{array}$$
        $$\begin{array}{rcl}\iota_{\RR}: \GL(n,\E)\rtimes \Aut(\F)) & \longrightarrow & (((\F^*)^{\ell} \times \GL(\mathbf{n},\E)) \rtimes \mathcal S_{\lambda(\mathbf{n})})\rtimes \Aut(\F) \\
        (M,\theta_2) & \longmapsto & (\mathbf{1}, \underbrace{M, \ldots, M}_{\ell \mbox{ \tiny{times} } }, \mathrm{id}, \theta_2).
    \end{array}$$
  \end{proposition}

\begin{proof}
 Consider a nonzero codeword $c$ in $\C_1\otimes \C_2$ as a matrix in $\F^{\ell\times N}$, where each column belongs to $\C_1$ and each row belongs to $\C_2$. We denote the rows of $c$ as $c^{(1)}, \ldots, c^{(\ell)}\in\C_2$ and the columns of $c$ as $c_{(1)},\ldots,c_{(N)}\in\C_1$. Let $\psi_{\HH}:=(\mathbf{a},\theta_1,\pi)\in\Aut_{\HH}(\C_1)$ and let $\psi_{\RR}:(M,\theta_2)\in \Aut_{\RR}(\C_2)$ be arbitrary. We only need to show that $\iota_{\HH}(\psi_{\HH}),\iota_{\RR}(\psi_{\RR}) \in \Aut_{\SR}(\C_1 \otimes\C_2)$.
 
 First, consider  $\iota_{\HH}(\psi_{\HH})$ acting on $c$. Writing 
 $$ c=\left(\begin{array}{c|c|c}c_{(1)} & \ldots & c_{(N)} \end{array} \right),$$
 we have that $\iota_{\HH}(\psi_{\HH})$ acts on each $c_{(i)}$ as $\psi_{\HH}$. Since $\psi_{\HH}\in\Aut_{\HH}(\C_1)$, then each column of $\iota_{\HH}(\psi_{\HH})(c)$ still belongs to $\C_1$, and thus $\iota_{\HH}(\psi_{\HH})(c)\in\C_1\otimes \C_2$. For the arbitrariness of the codeword $c$, we deduce that $\iota_{\HH}(\psi_{\HH})\in\Aut_{\SR}(\C_1 \otimes \C_2)$.
 
 Now consider  $\iota_{\RR}(\psi_{\HH})$ acting on $c$. Writing 
 $$ c=\left(\begin{array}{c}\;\;c^{(0)}\;\; \\ \hline \vdots \\ \hline c^{(\ell-1)} \end{array} \right),$$
 we have that $\iota_{\RR}(\psi_{\RR})$ acts on $c$ as 
 $$ \iota_{\RR}(\psi_{\RR})(c)=\left(\begin{array}{c}\;\;\theta_2(c^{(0)})M \;\; \\ \hline \vdots \\ \hline \theta_2(c^{(\ell-1)})M \end{array} \right).$$
 Therefore, $\iota_{\RR}(\psi_{\RR})$ naturally acts on each row of $c$ as $\psi_{\RR}$, showing that $\iota_{\RR}(\psi_{\RR})(c)\in \C_1\otimes \C_2$. For the arbitrariness of the codeword $c$, we deduce that $\iota_{\RR}(\psi_{\RR})\in\Aut_{\SR}(\C_1 \otimes \C_2)$.
\end{proof}

We first give a group-theoretic proof of the fact that the tensor product of a cyclic code and a skew-cyclic code is cyclic-skew-cyclic. 

   \begin{proposition}\label{prop:product_cyclic_skewcyclic}
   Let $\C_1$ be a cyclic code of length $\ell$ over $\F$ and let $\C_2$ be a skew-cyclic code of length $N$ over $\F$. Then, $\C_1\otimes \C_2$ is (equivalent to) a cyclic-skew-cyclic code.
  \end{proposition}

\begin{proof}
 By assumption, $\C_1$ is cyclic, so $(\mathbf{1},\pi_1,\mathrm{id}) \in \Aut_{\HH}(\C_1)$, where $\pi_1=(1\,2\,\cdots\,\ell)$ is the right-shift operator on $\F^\ell$. Furthermore, $\C_2$ is skew-cyclic with respect to $\theta$, hence $(P, \theta) \in \Aut_{\RR}(\C_2)$, where $P$ is the $N\times N$ permutation matrix associated to  $(1\,2\,\cdots\,N)$. 
 By Proposition \ref{prop:tensor_product_autgroup}, we have that $\iota_{\HH}((\mathbf{1},\pi_1,\mathrm{id})),\iota_{\RR}((P, \theta))\in \Aut_{\SR}(\C_1\otimes \C_2)$. However, it is easy to see that 
 $$ \iota_{\HH}((\mathbf{1},\pi_1,\mathrm{id}))=\rho, \qquad \iota_{\RR}((P, \theta))=\phi,$$
 where $\rho$ and $\phi$ are given in \eqref{eq:rho} and \eqref{eq:phi}, respectively. Hence, $\C_1\otimes \C_2$ is a cyclic-skew-cyclic code.
\end{proof}

    In the case described in Proposition \ref{prop:product_cyclic_skewcyclic}, in principle one might have difficulties in understanding the structure of the code $\C_1\otimes \C_2$ starting from the structure of $\C_1$ and $\C_2$. Also the proof does not really help, since it is based on group theory. However, as mentioned in Remark \ref{rem:PIR}, when $\mathrm{char}(\F)$ does not divide $\ell$, it was shown in \cite{martinez2020sum} that the ring
    $$\mathcal R':= \faktor{\left(\faktor{\F[x]}{(x^\ell - 1)}\right)[z;\theta]}{(z^m-1)}$$
    is a principal left ideal ring. The good news is that in this case we can determine the generator polynomial of $\nu(\C_1\otimes \C_2)$.
    
    First, we introduce the standard polynomial representations for cyclic and skew-cyclic codes. 
    We define the map $\mu_{\HH} :\F^\ell \rightarrow \mathcal S'$ as
    $$ \mu_{\HH}(u_0,\ldots,u_{\ell-1}) = \sum_{i=0}^{\ell-1} u_ix^{i},$$
    and the map $\nu_{\RR} :\F^N \rightarrow \mathcal S$ as
    $$ \nu_{\RR}(v_0,\ldots,v_{N-1}) = \sum_{j=0}^{N-1} v_jz^{j}.$$
    
    Here we give the explicit description of the generator polynomial of 
   $\nu(\C_1\otimes \C_2)\subseteq \mathcal R'$
   in terms of the generator polynomials of $\mu_{\HH}(\C_1)$ and $\nu_{\RR}(\C_2)$. 
   
   \begin{theorem}\label{thm:product_cyclic_skewcyclic} Suppose that $\ell$ and $\mathrm{char}(\F)$ are coprime.  
    Let $\C_1$ be an $[\ell,k_1]_{\F}$ cyclic code with $\mu_{\HH}(\C_1)=(f_1(x))$ and let $\C_2$ be an $[N,k_2]_{\F/\E}$ skew-cyclic with $\nu_{\RR}(\C_2)=(f_2(z))$. 
    Then, 
    $\nu(\C_1\otimes \C_2)=(f_1(x)f_2(z))$.
   \end{theorem}
    
    \begin{proof}
    Let $g(x,z):=f_1(x)f_2(z)$ and define $I:=(g(x,z))\subseteq \mathcal R'$. First, observe that $\deg_x(g)=\deg_x(f_1)=\ell-k_1$ and $\deg_z(g)=\deg_z(f_2)=N-k_2$. Now, consider the set
    $$\mathcal P:=\{ x^iz^j g(x,z)\,: \, 0\leq i \leq k_1-1, 0 \leq j \leq k_2-1 \},$$
    which is clearly contained in $I$. Moreover, $\nu^{-1}(\mathcal P)$ is a set of $\F$-linearly independent vectors in $\F^{\ell N}$: to see that, choose any term ordering on the monomials $\{x^iz^j \,:\, i,j \in \N \}$. The elements of $\mathcal P$ have all distinct leading monomials, and one can immediately deduce that they are $\F$-linearly independent. Therefore, 
    $$\dim_{\F}(\nu^{-1}(I))\geq |\mathcal P|=k_1k_2.$$
    Now it is enough to show that $\nu^{-1}(I)\subseteq \C_1\otimes \C_2$ to deduce that they coincide. In particular, we only need to show that $\nu^{-1}(g(x,z)) \in \C_1\otimes \C_2$. If we write
    $$f_1(x)=\sum_{i=0}^{\ell-1} a_ix^{i}, \qquad f_2(z)=\sum_{j=0}^{N-1} b_jz^{j},$$
    then
    $$ f_1(x)f_2(z)=\sum_{i,j}a_ib_jx^iz^j$$
    and the matrix representation of $\nu^{-1}(f_1(x)f_2(z))$ is
    $$ \left(\begin{array}{c|c|c} b_0a^\top & \ldots & b_{N-1}a^\top \end{array} \right)=\left(\begin{array}{c} a_0b \\ \hline \vdots \\ \hline  a_{\ell-1} b \end{array} \right)=a^\top \otimes b,$$
    where $a=(a_0,\ldots,a_{\ell -1}), b=(b_0,\ldots, b_{N-1})$. Thus, it is clear that $\nu^{-1}(f_1(x)f_2(z))\in \C_1\otimes \C_2$.
    \end{proof}

    \section{Product Bounds for Cyclic and Skew-Cyclic Codes}\label{sec:product_bounds}
We can now determine the defining set of the tensor product of a cyclic codes and a cyclic-skew-cyclic code, under a certain assumption on the first code.  We introduce the following notation. Let $X,Y$ be two sets, and $A\subseteq X$, $B\subseteq Y$. We denote by $A\uplus B$ the set
$$A\uplus B:=(A\times Y)\cup(X\times B)=\{(x,y) \in X \times Y \mid x \in A \mbox{ or } y \in B\}.$$

We recall that, for an $[\ell,k]_\F$ cyclic code $\C_1$ such that $\mu_{\HH}(\C_1)=(f_1(x))$, the defining set is  
$$ T_{\C_1}^{\HH}=\{ a \in \LL \mid a^\ell=1, f_1(a)=0 \}.$$

Furthermore, for an $[N,k]_{\F/\E}$ skew-cyclic code $\C_2$ such that $\nu_{\RR}(\C_2)=(f_2(z))$, the defining set is  
$$ T_{\C_2}^{\RR}=\{ \beta \in \LL^* \mid \Ev_\beta^\sigma(f_2(z))=0 \}.$$

\begin{theorem}\label{thm:product_defining_sets}
  Let $\C_1$ be an $[\ell,k_1]_{\F}$ cyclic code with $\mu_{\HH}(\C_1)=(f_1(x))$ and let $\C_2$ be an $[N,k_2]_{\F}$ skew-cyclic code with $\nu_{\RR}(\C_2)=(f_2(z))$. Moreover, assume that $f_1(x)\in\E[x]$. Then, 
  $$T_{\C_1\otimes \C_2}=T^{\HH}_{\C_1}\uplus T^{\RR}_{\C_2}.  $$
\end{theorem}

\begin{proof}
 Let $g(x,z)=f_1(x)f_2(z)$. By Theorem \ref{thm:product_cyclic_skewcyclic} we have that $\nu(\C_1\otimes \C_2)=(g(x,z))$ and
 $$T_{\C_1\otimes \C_2}=\{(a,\beta) \in\K\times \LL^\ast \mid a^\ell = 1, \Ev_{a,\beta}(g(x,z)) = 0\}.$$
 However, we have 
 \begin{align*}\Ev_{a,\beta}(g(x,z))&=\Ev_{\beta}^\sigma(\Ev_{a,z}(g(x,z)))=\Ev_{\beta}^\sigma(\Ev_{a,z}(f_1(x)f_2(z))) \\&= \Ev_{\beta}^\sigma(f_1(a)f_2(z))=f_1(a)\Ev_{\beta}^\sigma(f_2(z)),\end{align*}
 where the second to last equality follows from \eqref{eq:evaluation_ring_homo}.
 From this, we immediately conclude the proof.
\end{proof}

We are now going to combine the Roos bound and the Hartmann-Tzeng bound for cyclic-skew-cyclic codes obtained in Section \ref{sec:Roos_HT_bounds}, with the results on the minimum sum-rank distance of a product code  of Proposition \ref{prop:sumrank_product_code}. This allows  to derive new bounds on the minimum (Hamming) distance of cyclic codes and on the minimum (rank) distance of  skew-cyclic codes, and it will be done with the aid of Theorem \ref{thm:product_defining_sets}.

\begin{theorem}[Roos Product Bound]\label{thm:Roos_product}
  Let $\C_1$ be an $[\ell,k_1,\dd_{\HH}]_{\F}$ cyclic code whose generator polynomial belongs to $\E[x]$ and let $\C_2$ be an $[m,k_2,\dd_{\RR}]_{\F/\E}$ skew-cyclic code. Let $b, s,\delta, k_0,\dots, k_r$ be integers, such that $\gcd(\ell m,s)=1$, $k_i<k_{i+1}$ for $i=0,\dots,r-1$, $k_r-k_0 \leq \delta + r -2$. Let $a\in\K$ be a primitive $\ell$-th root of unity and $\beta$ be a normal element of $\LL/\K$. If
  $$ \{(a^{b+si+k_j}, \sigma^{si+k_j}(\beta))\in \K \times \LL^\ast \, : \,  0\leq i\leq \delta-2, 0\leq j\leq r\}\subseteq  T^{\HH}_{\C_1}\uplus T^{\RR}_{\C_2},$$
  then
     $$\dd_{\HH}\geq \left\lceil \frac{\delta+r}{\dd_{\RR}}\right\rceil \qquad \mbox{ and } \qquad \dd_{\RR}\geq \left\lceil\frac{\delta+r}{\dd_{\HH}}  \right\rceil. $$
\end{theorem}

\begin{proof}
 Let us consider the product code $\C_1\otimes \C_2$. By Theorem \ref{thm:product_cyclic_skewcyclic}, $\C_1\otimes \C_2$ is a cyclic-skew-cyclic code whose defining set is $T_{\C_1 \otimes \C_2}=T^{\HH}_{\C_1}\uplus T^{\RR}_{\C_2}$, by Theorem \ref{thm:product_defining_sets}. Using the Roos bound for sum-rank metric codes of Theorem \ref{thm:Roos_SR}, we hence deduce that 
 $$\dd_{\SR}(\C_1\otimes \C_2) \geq \delta+r.$$
 On the other hand, by Proposition \ref{prop:sumrank_product_code}, we also know that
 $$\dd_{\SR}(\C_1\otimes \C_2)=\dd_{\HH}(\C_1)\dd_{\RR}(\C_2).$$
 Combining this with the previous inequality, we obtain the desired bounds.
\end{proof}

We remark that clearly the two bounds of Theorem \ref{thm:Roos_product} are equivalent. However, we decided to explicitly state both in order to emphasize the fact that we can lower bound both the minimum Hamming distance of the cyclic code $\C_1$ and the minimum rank distance of the skew-cyclic code $\C_2$.
In particular, depending on the situation, we may design a suitable skew-cyclic (resp. cyclic) code to obtain new lower bounds on the minimum Hamming (resp. rank) distance of a given cyclic (resp. skew-cyclic) code.

We conclude this section deriving in the same way an Hartmann-Tzeng product bound. The result is based on Theorem \ref{thm:HT_SR} instead of Theorem \ref{thm:Roos_SR}, but the rest of its proof is completely analogous to the proof of Theorem \ref{thm:Roos_product}. For this reason we omit it. 

\begin{theorem}[Hartmann-Tzeng Product Bound]\label{thm:HT_product} 
  Let $\C_1$ be an $[\ell,k_1,\dd_{\HH}]_{\F/\E}$ cyclic code whose generator polynomial belongs to $\E[x]$ and let $\C_2$ be an $[m,k_2,\dd_{\RR}]_{\F}$ skew-cyclic code. Let $b,\delta, r,t_1,t_2$ be integers, such that $\gcd(\ell m,t_1)=1$, $\gcd(\ell m,t_2)<\delta$.  Let $a\in\K$ be a primitive $\ell$-th root of unity and $\beta$ be a normal element of $\LL/\K$. If
  $$\{(a^{b+it_1+st_2}, \sigma^{it_1+st_2}(\beta))\in \K \times \LL^\ast \, : \,  0\leq i\leq \delta-2, 0\leq s\leq r\}\subseteq  T^{\HH}_{\C_1}\uplus T^{\RR}_{\C_2},$$
  then
     $$\dd_{\HH}\geq \left\lceil \frac{\delta+r}{\dd_{\RR}}\right\rceil, \qquad \dd_{\RR}\geq \left\lceil\frac{\delta+r}{\dd_{\HH}}  \right\rceil. $$
\end{theorem}

\bigskip

\section*{Acknowledgments}

The work of Gianira N. Alfarano is supported by  Swiss National Science Foundation grant n. 188430.
The work of Antonia Wachter-Zeh is supported by the German Research Foundation (DFG) under Grant No. WA3907/1-1. The work of F. J. Lobillo is supported by SRA (State Research Agency / 10.13039/501100011033) under Grant No. PID2019-110525GB-I00.
\bigskip
\bibliographystyle{abbrv}
\bibliography{references.bib}

\end{document}